\documentclass{llncs}
\usepackage{amsmath}

\usepackage[latin1]{inputenc}
\usepackage[english]{babel}
\usepackage{microtype}
\usepackage{multicol}
\allowdisplaybreaks[1]

\usepackage[utopia]{mathdesign}
\usepackage{bm}
\usepackage{shuffle}

\newcommand{\N}{\mathbb{N}}
\newcommand{\shf}{\shuffle}
\newcommand{\EXTRA}[1]{}

\titlerunning{Square-free Shuffles of Words}
\authorrunning{Tero Harju \and Mike M\uller}
\begin{document}

\mainmatter
\title{Square-Free Shuffles of Words}

\author{Tero Harju\inst{1} \and Mike M\"uller\inst{2}\thanks{Supported by the DFG grant 582014}}
\institute{
         Department of Mathematics\\
         University of~Turku, Finland\\
\email{harju@utu.fi} \and
Institut f\"ur Informatik\\
Christian-Albrechts-Universit{\"a}t zu Kiel, Germany\\
\email{mimu@informatik.uni-kiel.de} }

\maketitle

\begin{abstract}
Let $u \shuffle v$ denote the set of all shuffles of the words $u$ and $v$.
It is shown that for each integer $n \geq 3$ there exists a square-free ternary word $u$
of length $n$ such that $u\shuffle u$ contains a square-free word. This property is then shown to also hold for
infinite words, i.e.,  there exists an infinite square-free word~$u$ on three letters such
that~$u$ can be shuffled with itself to produce an infinite square-free word $w \in u
\shuffle u$.
\end{abstract}
%



\section{Introduction}

Let $u$ and $v$ be words over a finite alphabet $\Sigma$. We let
\[
u \shuffle v = \{u_1v_1 \cdots u_nv_n \mid
u = u_1u_2 \cdots u_n, \ v=v_1v_2 \cdots v_n \
(u_i, v_i \in \Sigma^*)\}
\]
be the set of all \emph{shuffles} of $u$ and $v$. This definition extends to infinite words in
a natural way. In this paper we consider avoidance of  repetitiveness among the shuffled
words. This topic was studied, e.g., by Prodinger and Urbanek~\cite{ProdingerUrbanek} in
1979 while they considered squares in the shuffles of two words; see also~Currie and
Rampersad~\cite{CurrieRampersad} and Rampersad et al.~\cite{RampersadSW}.

In Charlier et al.~\cite{CKPZ} the authors considered \emph{self-shuffling} of infinite words.
In this problem setting shuffling is applied to an infinite word $w$ such that $w \in w
\shuffle w$. In~\cite{CKPZ} a short and elegant proof is given for the fact that the Fibonacci
word can be self-shuffled, and a longer proof is provided for the self-shuffling property of
the Thue--Morse word.

In this paper we are interested in finite and infinite square-free words~$w$ that can
be obtained by shuffling two square-free words, $u$ and~$v$, i.e., $w \in u \shuffle v$.
We show first that for each integer $n \geq 3$ there exists a square-free ternary word~$u$
such that $u\shuffle u$ contains a square-free word (of length $2n$).
Next it is shown
that there exists an infinite square-free word~$u$ on three letters such that~$u$ can be
shuffled with itself to produce an infinite square-free word $w \in u \shuffle u$.
The existence of self-shuffled words remains an open
problem, but we are able to show that
there are infinite square-free words $u$ and $w$ that can be shuffled to produce $u$
again, i.e., $u \in u \shuffle w$.

\goodbreak
\section{Preliminaries}

For each positive integer $n$, let $\Sigma_n = \{0,1, \ldots, n-1\}$ be a fixed alphabet of~$n$
letters. We shall be needing only small alphabets.

Let $u_0, u_1 \in \Sigma^*$ be words over an alphabet $\Sigma$ and let $\beta \in
\Sigma_2^*$ be a binary word of length $|\beta|=|u_0|+|u_1|$, called a
\emph{conducting sequence}, such that the number of the letter $i \in \Sigma_2$ in $\beta$
is equal to the length $|u_i|$. While forming the shuffle
\[w= u_0 \shuffle_\beta u_1\]
of $u_0$ and $u_1$ \emph{conducted by} $\beta$, at step $i$ the sequence $\beta$ will
choose the first unused letter from $u_0$ if $\beta(i)=0$ or the first unused letter from
$u_1$ if $\beta(i)=1$. That is, the $i$th letter $w(i)$  of $w$ becomes defined by
\[
w(i)=
u_{\beta(i)}(j) \  \text{ where  }  j=
\mathrm{card}\{k \mid \beta(k)=\beta(i) \ \text{ for } \ k=1,2, \ldots, i\}\,.
\]
This definition can be extended to infinite words $u,v \in \Sigma^\N$ in a natural way. Now
$u \shf_\beta v \in \Sigma^\N$ is an infinite word obtained by shuffling $u$ and $v$
conducted by the sequence $\beta \in \Sigma_2^\N$, where one requires that $\beta$
contains infinitely many occurrences of both $0$ and~$1$.

\begin{example}
Let $u=0102$ and $v=1201$ be two ternary words of length four, and let $\beta=00101110$
be their conducting sequence. Then $u\shuffle_\beta v = 01102012$.\qed
\end{example}

\begin{example}
A shuffled word $w=u \shf_\beta u$ can be obtained in more than one way from a single
word  $u$ using different conducting sequences. To see this, let, e.g.,  $u=012102010212$
and choose
\begin{align*}
\beta_1=&000000000001111111101111\\
\beta_2=&000000110100100111101111\,.
\end{align*}
Then  $w=u \shf_{\beta_1} u =u \shf_{\beta_2} u =
012102010210121020120212$.
In this example the words $u$ and $w$ are even square-free.
\qed\end{example}

A finite or infinite word $w$ over an alphabet $\Sigma$ is \emph{square-free} if it does not
have factors of the form $u^2 =uu$ for nonempty words $u$.  A hundred years ago
Axel~Thue constructed an infinite square-free word on three letters. One example of such a
word, see Hall~\cite{Hall} or Lothaire~\cite{Lothaire}, is the fixed point of the \emph{Hall
morphism}
\[
\tau(0)=012, \quad \tau(1)=02, \quad \tau(2)=1\,.
\]
The iteration of $\tau$ on $0$ gives a square-free word
\begin{equation}\label{thue}
t=012021012102012021020121012 \cdots \,.
\end{equation}

Let $h\colon \Sigma^* \to \Delta^*$ be a morphism of words. It is $k$-\emph{uniform} if $|h(a)|=k$
for all $a \in \Sigma$. Also,  $h$ is \emph{square-free}, if it preserves square-freeness of
words, i.e., if $v\in \Sigma^*$ is square-free, then so is the image $h(v) \in \Delta^*$.

A (monoid) morphism $h\colon \Sigma^* \to 2^{\Delta^*}$ from word monoid  to a power monoid of words
is called a \emph{substitution}. A substitution~$h$ is said to be \emph{square-free} if
for all square-free  $w\in \Sigma^*$, the image $h(w)$ consists of square-free words only.

The following result is due to Crochemore~\cite{Crochemore} improving a result of Bean et
al.~\cite{Bean}; see also Lothaire~\cite{Lothaire}.

\begin{theorem}\label{Crochemore}
A morphism $h\colon \Sigma^* \to \Delta^*$ of words is square-free if it preserves square-freeness of
words of length
\[
\max\left( 3, \left\lceil \frac{M-3}{m}+1\right\rceil \ \right)\,,
\]
where $M=\max( \, |h(a)| \ : \ a \in \Sigma)$ and $m=\min( \,|h(a)| \ : \  a \in \Sigma)$. In
particular, if $h$ is uniform, then it is square-free if it preserves square-freeness of words
of length~3.
\end{theorem}


\section{Shuffles of Finite Words}

\begin{example}
The shuffled word $u \shf_\beta u$ can be square-free even if $u$ is not so. For instance, if
$u=(012)^2$ then the following shuffled words are square-free:
\begin{align*}
    u \shf_{\beta_0} u= &012010212012, \text{ where } \beta_0= 000001011111\\
    u \shf_{\beta_1} u= &010201210212, \text{ where } \beta_1= 001001101011\\
    u \shf_{\beta_2} u= &010210120212, \text{ where } \beta_2= 001010011011\,.
\end{align*}
\qed\end{example}

In~\cite{harju2} the first author asked whether for each $n \geq 3$, there exists a
square-free word $u$ of length~$n$ such that $u \shf_\beta u$ is square-free for some
$\beta$. We give an affirmative answer to this question after a short technical lemma that
will be used in our construction:

\begin{lemma}\label{lem:sub}
    The substitution $h : \Sigma_3^* \rightarrow 2^{\Sigma_3^*}$, defined by
    \begin{align*}
        h(0) &\in \{01202120102120210, 012021020102120210\} \\
        h(1) &\in \{12010201210201021, 120102101210201021\} \\
        h(2) &\in \{20121012021012102, 201210212021012102\}
    \end{align*}
is square-free. 
\end{lemma}

\begin{proof} 

Note that the words in $h(a)$, $a \in \Sigma_3$, have lengths 17 and 18.
The substitution $h$ has the following three properties, which are easy to check:
\begin{enumerate}
    \item No image of a letter appears properly inside any image of a word of length~$2$,
          i.e., $h(ab) \cap xh(c)y = \varnothing$ for $a,b,c \in \Sigma_3$ and $x,y \in
          \Sigma_3^*$, if $x,y \neq \varepsilon$.
    \item No image of a letter is a prefix of an image of another letter.
    \item For $a,b \in \Sigma_3$ with $a \neq b$, $a' \in h(a)$ and $b' \in h(b)$ end with a different letter.
\end{enumerate}
Assume towards a contradiction that $w$ is square-free, but $w' \in h(w)$ contains a square
$uu$. A simple inspection shows that $h$ produces no square $uu$ with $|u| \leq 52$, as
this would be contained inside the image of a square-free word of length $8$, and we can check all of
these. Therefore, assume that $|u| > 52$, so that $u$ contains at least two full images of a
letter under $h$. Let now
\[
u = s_1w_1'\cdots w_j'p_1 = s_2w_{j+2}'\cdots
w_n'p_2
\]
here $p_1s_2 \in h(w_{j+1})$, $w_i' \in h(w_i)$ for all $1 \leq i \leq n$, and $s_1$ is a suffix of some $w_0' \in h(w_0)$ and $p_2$ is a prefix
of some $w_{n+1}' \in h(w_{n+1})$. Also, $s_1, s_2 \neq \varepsilon$.

If $|s_1| > |s_2|$, then $w_{j+2}' \in h(w_{j+2})$ appears properly inside $w_0'w_1' \in h(w_0w_1)$; a
contradiction to property (1). The same situation appears if $|s_1| < |s_2|$, then
$w_1'$ occurs properly inside $w_{j+1}'w_{j+2}'$. If $|s_1| = |s_2|$, then by iterated
application of property (2), we get that $w_1' = w_{j+2}', \ldots, w_j' = w_n'$. As
(2) also shows that $h$ is injective, we have $w_1 = w_{j+2}, \ldots, w_j = w_n$.
Furthermore, by property (3) we have $w_0 = w_{j+1}$, and thus a square $(w_0\cdots
w_j)^2$ in $w$. This proves the claim.\qed
\end{proof}

The following result is clear, and it follows the idea of Charlier et al.~\cite{CKPZ}.

\begin{lemma}
    Let $u$ be a square-free word that can be shuffled to a square-free word
    $u \shf_\beta u = u_1u'_1u_2u'_2\cdots u_nu'_n$, and let $h$ be a square-free morphism or substitution,
    then $h(u)$ can also be shuffled with itself to get the square-free word
    $h(u \shf_\beta u) = h(u_1)h(u'_1)h(u_2)h(u'_2)\cdots h(u_n)h(u'_n)$.
\end{lemma}

\begin{theorem}
For each $n \geq 3$, there exists a square-free word $w_n \in \Sigma_3^*$ of length $n$
such that $w_n \shf_\beta w_n$ is square-free for some $\beta$.
\end{theorem}

\begin{proof}
First of all, if $u \in \Sigma_3^*$ is any non-empty square-free word, then $u34 \in
\Sigma_5^*$ is also square-free. Furthermore, for $\beta = 0^{|u|+1}1^{|u|}011$, we have
$u34 \shf_\beta u34=u3u434$, which is obviously square-free as well. Thus, there exist
square-free words $v \in \Sigma_5^*$ of length $n$ such that $v \shf_{\beta_n} v$ is
square-free for $\beta_n = 0^{n-1}1^{n-2}011$ for each $n \geq 3$.

We map these to words in $\Sigma_3^*$ of length $18n$ for all $n \geq 3$ using an
$18$-uniform square-free morphism from $\Sigma_5^*$ to $\Sigma_3^*$ due to
Brandenburg \cite{Brandenburg}. Applying the substitution $h$ from Lemma \ref{lem:sub}
to the result, we construct words with the desired property of all integer lengths in the
intervals $[18\cdot17\cdot n, 18\cdot 18 \cdot n]$ for all $n \geq 3$. We notice that for $n \geq 17$, the
intervals obtained from $n$ and $n+1$ intersect. Therefore this construction produces
square-free words in $\Sigma_3^*$ of all lengths $\geq 18 \cdot 17 \cdot 17 = 5202$, that
can be shuffled with themselves to get a square-free word. What is more, Brandenburg
\cite{Brandenburg} found also a $22$-uniform square-free morphism from $\Sigma_5^*$ to
$\Sigma_3^*$, and there are $k$-uniform morphisms of that kind for $k \in \{19, 23, 24\}$
as well:
    \begin{align*}
        h_{19}(0) &= 0102012021020121012 &  h_{23}(0) &= 01020120210120102120121 \\
        h_{19}(1) &= 0102012021201021012 &  h_{23}(1) &= 02120102012021020121012 \\
        h_{19}(2) &= 0102012021201210212 &  h_{23}(2) &= 02120102012102120121012 \\
        h_{19}(3) &= 0102012102120121012 &  h_{23}(3) &= 02120102101210201021012 \\
        h_{19}(4) &= 0210201202120121012 &  h_{23}(4) &= 02120102101210212021012
    \end{align*}
    \begin{align*}
        h_{24}(0) &= 010201202101201020120212 \\
        h_{24}(1) &= 010201210120102101210212 \\
        h_{24}(2) &= 010201210120210201210212 \\
        h_{24}(3) &= 010201210201021201210212 \\
        h_{24}(4) &= 010210120102120210120212
    \end{align*}
Square-freeness of these morphisms is proven using Theorem \ref{Crochemore}. With the
construction above and these morphisms, we find square-free words $w_n \in \Sigma_3^*$
that can be shuffled with themselves to get a square-free word for each length $n$, where
\begin{align*}
    n \in \bigcup_{\substack{k \in \{18, 19, 22, 23, 24\} \\ n\geq 3}}[k\cdot17\cdot n, k\cdot18\cdot n].
\end{align*}

Furthermore, Currie \cite{Currie} constructed $k$-uniform square-free morphisms for all
$k \geq 11$, except for $k \in \{14,15,16,20,21,22\}$. Using these we construct square-free
words with the desired property for all $n$ that are divisible by some $d \geq 11$ and $d
\notin \{14,15,16,20,21,22\}$, and $n/d \geq 3$ from the word $w_{n/d}$.

Combining all these results, there are only $335$ values left, for which a square-free word
$w_n \in \Sigma_3^*$ that can be shuffled with itself to get a square-free word must be
explicitly constructed. Words of the lengths we found by a computer search, see the table in
the Appendix.\qed
\end{proof}

\goodbreak

\section{Shuffling Infinite Square-Free Words}

After some preliminaries and examples around the problem we shall prove the following
theorem.

\begin{theorem}\label{main}
There exists an infinite square-free word $u$ on three letters and a conducting sequence
$\beta \in \Sigma_2^\N$ such that $u \shf_\beta u$ is square-free.
\end{theorem}

\begin{proof} \ We observe first that if $\beta_1$ and
$\beta_2$ are conducting sequences with $\beta_1$ finite in length and containing equally
many 0's and 1's, then
\begin{equation}\label{concat}
u \shuffle_{\beta_1\beta_2} u= \left(u_1\shuffle_{\beta_1} u_1\right)\left(u_2\shuffle_{\beta_2} u_2\right)
\ \text{ where }
u=u_1u_2 \text{ with } 2\cdot|u_1|=|\beta_1|.
\end{equation}
The words to be shuffled will be the images of the 12-uniform morphism $\rho\colon
\Sigma_4^* \to \Sigma_3^*$ defined by
\begin{align*}
\rho(0)=&010210120212\\
\rho(1)=&012101202102\\
\rho(2)=&012102010212\\
\rho(3)=&012102120102\,.
\end{align*}
Each word $\rho(i)$ is square-free, but \emph{the
morphism~$\rho$ is not square-free}. Indeed, $\rho(12),\rho(20)$ and $\rho(30)$
contain squares. For instance, $\rho(20)=01210\cdot  (201021)^2 \cdot 0120212$. For this
reason, we need a morphism $\alpha$ that will fix this problem. It will be defined below.

Each of the words $\rho(i)$, for  $i=0,1,2,3$, can be shuffled to obtain a square-free word
$\sigma(i)=\rho(i)\shuffle_{\beta_i}\rho(i)$ as seen in Table~\ref{table2}.

\begin{table}[htb]
\begin{center}
\begin{tabular}{lcl}
$\sigma(0)=010210120102120210120212$, &\quad& $\beta_0 = 000000001100001111111111$\\
$\sigma(1)=012101202101210201202102$, &\quad& $\beta_1 = 000000000011110011111111$\\
$\sigma(2)=012102010210121020120212$, &\quad& $\beta_2 = 000000110100100111101111$\\
$\sigma(3)=012102120102101202120102$, &\quad& $\beta_3 = 000000001101001011111111$\,.
\end{tabular}
\end{center}
\caption{Square-free words $\sigma(i)=\rho(i)\shuffle_{\beta_i}\rho(i)$ of length 24. The  column on the
right  shows the conducting sequences} \label{table2}
\end{table}

Next, let the uniform morphism $\alpha\colon \Sigma_3^* \to \Sigma_4^*$ be defined by
\begin{align*}
\alpha(0)=&1013\\
\alpha(1)=&1023\\
\alpha(2)=&1032\,.
\end{align*}
Notice that, for any $w \in \Sigma_3^*$, the image $\alpha(w)$ avoids the
`forbidden' words $12,20$ or $30$. Also, the word $10$ occurs in $\alpha(w)$
only as a prefix of each $\alpha(a)$ for $a\in \Sigma_3$. It is then easy to prove, and it also
follows by applying Theorem~\ref{Crochemore}, that the morphism $\alpha$ is square-free.

Finally, we combine the above morphisms to obtain $B, S\colon \Sigma_3^* \to \Sigma_3^*
$ by letting
\begin{align*}
B(i)&=\rho \alpha(i)\\
S(i)&=\sigma \alpha(i)
\end{align*}
for $i=1,2,3$. The images of the words $B(i)$ are:
\begin{align*}
B(0)=&012101202102010210120212012101202102012102120102\\
B(1)=&012101202102010210120212012102010212012102120102\\
B(2)=&012101202102010210120212012102120102012102010212\,.
\end{align*}
The lengths of these words are 48. The images of the shuffled words are of length 96:
\begin{align*}
S(0)=&01210120210121020120210201021012010212021012021201210\\
\cdot &1202101210201202102012102120102101202120102\\
S(1)=&01210120210121020120210201021012010212021012021201210\\
\cdot &2010210121020120212012102120102101202120102\\
S(2)=&01210120210121020120210201021012010212021012021201210\\
\cdot &2120102101202120102012102010210121020120212\,.
\end{align*}
Now, Theorem~\ref{Crochemore} and a computer check verify that the morphisms $B$ and
$S$ are square-free. Hence if $w=i_1i_2 \cdots$ is an infinite square-free ternary word  in
$\Sigma_3^\N$, then both $B(w)$ and $S(w)$ are square-free. By the constructions of $B$
and $S$, we have
\begin{equation}\label{shuf}
S(i)=\sigma\alpha(i) = \rho\alpha(i)\shuffle_{\beta_{\alpha(i)}}\rho\alpha(i)
=B(i)\shuffle_{\beta_{\alpha(i)}}B(i) \ \text{ for each } i \in \Sigma_3\,.
\end{equation}
Then
\[
S(w) = S(i_1)S(i_2) \cdots = \left(B(i_1)\shuffle_{\beta_{\alpha(i_1)}} B(i_1)\right)
\left( B(i_2) \shuffle_{\beta_{\alpha(i_2)}} B(i_2)\right) \, \cdots\,,
\]
and  inductively using \eqref{concat}, 
we find that $\beta = \beta_{\alpha(i_1)}\beta_{\alpha(i_2)} \cdots$ is a conducting
sequence such that $S(w)=B(w) \shuffle_{\beta} B(w)$. This proves the claim.
\end{proof}\qed

We also observe that the words $B(0), B(1)$ and $B(2)$  have equally many each of the
letters, and therefore these words are Abelian equivalent.

\begin{corollary}\label{cor}
There exist infinite square-free ternary words that are Abelian periodic.
\end{corollary}

A simpler solution to Corollary~\ref{cor} was given in~\cite{Harju}, where an infinite
square-free ternary word was constructed that has Abelian period  equal to three.

\section{Infinite ``Almost'' Self-Shuffling Words}
In the previous section, we studied words $u$ that can be shuffled with themselves to get
another square-free word $w = u \shf_\beta u$ for some $\beta$. Now, we alter the
problem slightly, and study whether a square-free word $u$ can be shuffled with another
square-free word $w$ to get $u$. Since $|u \shf_\beta w| = |u|+|w|$ for all conducting
sequences $\beta$, this is not possible for finite words, unless $w = \varepsilon$, which is
trivial. The following theorem shows that this is however possible for infinite words:

\begin{theorem}\label{uinuSw}
    There exist infinite square-free words $u, w \in \Sigma_3^\N$
    and a conducting sequence $\beta$, such that $u = u \shf_\beta w$.
\end{theorem}
\begin{proof}
Let $u$ be the infinite fixed point $h^\omega(0)$ of the following $18$-uniform morphism
$h$:
\begin{align*}
    h(0) &= 012021\bm{0}20102120210 \\
    h(1) &= 120102\bm{1}01210201021 \\
    h(2) &= 201210\bm{2}12021012102.
\end{align*}
Both $h$ and the morphism $h'$ that is obtained by deleting the bold-face letters from
every image are square-free by Theorem \ref{Crochemore}. Furthermore, the sequence of
bold-face letters in $u$ equals $u$. Let now $w$ be the word that consists of the
non-boldface letters. As $w = h'(u)$, it is square-free, and furthermore $u = u \shf_\beta w$
for $\beta = (0^{6}10^{11})^\omega$.\qed
\end{proof}

\goodbreak
\section{Open questions}

\begin{problem}\label{prob:1}
Which square-free words $u$ can be shuffled to obtain a square-free word $w=u \shuffle_\beta u$?
\end{problem}

\begin{problem}
Characterize the words $u$ that can be shuffled to a unique square-free word $u \shuffle_\beta u$?
\end{problem}

\begin{problem}
Which words $w$ can be obtained in more than one way from a single word $u$ using
different conducting sequences?
\end{problem}

\begin{example}
The same square-free word can be shuffled to produce different square-free words; see
Table~\ref{table8}. As one can see there $u=01021201$ gives rise to three square-free
words $\beta_i(u)$, but, e.g., $u=01201021$ gives rise to  a single one. The rest of the
square-free words of length eight with prefix $01$ do not shuffle to any square-free word.
\qed\end{example}
\begin{table}[thp]
\begin{center}
\begin{tabular}{lll}
Word $u$ &Shuffled $u \shuffle_\beta u$& Conducting $\beta$\\ \hline
01021201&0102120102012101&0000001111011011\\
&0102120102101201&0000001111100111\\
&0102101201021201&0000011000111111\\
01201021&0102101201020121&0010100111001011\\
01202101&0120210120102101&0000001110011111\\
&0120102101202101&0001100000111111\\
&0102012101202101&0010010000111111\\
01202102&0102120210201202&0010110001101011\\
01202120&0120210201202120&0000001001111111\\
01210120&0121012010210120&0000000110111111\\
01210201&0121020102101201&0000001101110111\\
&0120102012101201&0001000011110111\\
&0120102101210201&0001000100111111\\
\end{tabular}
\end{center}
\caption{Square-free words $u$ of length 8  with a prefix $01$ having a square-free shuffle
$u \shuffle_\beta u$. (Missing items for $u$ mean that they are the same as in the above
line)}\label{table8}
\end{table}

\begin{example}
Square-free words $w$ that are shuffles of $w=u \shuffle_\beta u$ of square-free $u$ seem to be
relatively few compared to the number of all square-free words.
This is natural since if $w=u \shuffle_\beta u$ then $w$ must contain an even parity of
each of the letters, and therefore a square-free $wv$ cannot be obtained as a shuffle
for any $v$ of length $|v| \le 5$.
Also, the number of
different square-free words $u$ for which there exists $\beta$ such that $u \shuffle_\beta u$ is
square-free is much lower. Table~\ref{enumeration} gives values for small lengths of~$w$.
\qed\end{example}

\begin{table}
\begin{center}
\begin{tabular}{|c|c|c|c|c|c|c|c|c|} \hline
\ L \ & \ square-free \ & \ $u \shuffle_\beta u$ \ & \ $u$ \ &\quad& \ L \ & \ square-free \ & \
$u \shuffle_\beta u$ \ & \ $u$ \ \\ \hline
4&18&0&0
&&
6&42&6&6\\
8&78&12 & 6
&&
10&144&30 & 12\\
12&264&24 & 18
&&
14&456&42 & 30\\
16&798&78 & 42
&&
18&1392&138 & 36 \\
20&2388&228 & 54
&&
22&4146&396 & 138\\
24&7032&588 &  168
&&
26&11892&1008&234 \\  \hline
\end{tabular}
\end{center}
\caption{%
The table gives the numbers of ternary square-free $w=u \shuffle_\beta u$ of length $L$ where $u$ is
square-free of length $L/2$. The column for $u$ gives the number of different square-free
words of length $L/2$ that can be shuffled to obtain a square-free word of length $L$.}
\label{enumeration}
\end{table}

The converse of Problem~\ref{prob:1} reads as follows.

\begin{problem}
Which square-free words $w$ are shuffles of square-free words: $w \in u \shuffle u$?
\end{problem}

One might ask if the word $u$ in Theorem \ref{uinuSw} could be any square-free word $u$:
\begin{problem}[due to I. Petrykiewicz]
    For any square-free word $u \in \Sigma_3^\N$, does there exist a square-free word $w \in \Sigma_3^\N$,
    such that $u = u \shf_\beta w$ for some $\beta$?
\end{problem}

A similar question can be asked with respect to Theorem \ref{main}: For any square-free
word $u \in \Sigma_3^\N$, does there exist a $\beta$, such that $u \shf_\beta u$ is
square-free? Here, the answer is negative:

\begin{example}
    Let $u$ be the lexicographically smallest square-free word.
    This is certainly a Lyndon word. According to \cite{CKPZ}, for every conducting sequence $\beta$,
    the word $u \shf_\beta u$ is lexicographically strictly smaller than $u$ and thus not square-free.
    \qed
\end{example}

The next question involves self-shuffling.

\begin{problem}
Does there exist an infinite square-free word $w$ such that $w= w\shf_\beta w$ for some
infinite $\beta$?
\end{problem}

Note that the Hall word $t$ is an infinite Lyndon word, and thus, by~\cite{CKPZ} it is not
self-shuffled.

\begin{example}
According to~\cite{CKPZ}, no infinite aperiodic Lyndon word can be self shuffled. However, finite
square-free Lyndon words can be shuffled to obtain other square-free Lyndon words. For
this, consider $u=01202102$. It can be shuffled to obtain $w=0102120210201202$ by using
the conducting sequence $\beta=0010111110010100$. Note, however, that $w < u$ in the
lexicographic ordering. \qed\end{example}

It was recently shown \cite{Buss,Rizzi} that given a word $w$, it is generally NP-complete
to decide if there is a word $u$ such that $w = u \shf_\beta u$ for some $\beta$. This
suggests the following question:
\begin{problem}
    Given a square-free word $w$, how hard is it to decide whether $w = u \shf_\beta u$ for some $\beta$ and some square-free word $u$?
\end{problem}

\bibliographystyle{plain}
\bibliography{bib-small}

\begin{thebibliography}{10}

\bibitem{Bean}
Dwight~R. Bean, Andrzej Ehrenfeucht, and George~F. McNulty.
\newblock Avoidable patterns in strings of symbols.
\newblock {\em Pacific J. Math.}, 85(2):261--294, 1979.

\bibitem{Brandenburg}
Franz-Josef Brandenburg.
\newblock Uniformly growing {$k$}th power-free homomorphisms.
\newblock {\em Theoret. Comput. Sci.}, 23(1):69--82, 1983.

\bibitem{Buss}
Sam Buss and Michael Soltys.
\newblock Unshuffling a square is {NP}-hard, 2012.
\newblock Submitted.

\bibitem{CKPZ}
Emilie Charlier, Teturo Kamae, Svetlana Puzynina, and Luca~Q. Zamboni.
\newblock Self-shuffling words.
\newblock In Fedor~V. Fomin, Rusins Freivalds, Marta~Z. Kwiatkowska, and David
  Peleg, editors, {\em ICALP (2)}, volume 7966 of {\em Lecture Notes in
  Computer Science}, pages 113--124. Springer, 2013.

\bibitem{Crochemore}
Max Crochemore.
\newblock Sharp characterizations of squarefree morphisms.
\newblock {\em Theoret. Comput. Sci.}, 18(2):221--226, 1982.

\bibitem{CurrieRampersad}
James Currie and Narad Rampersad.
\newblock Cubefree words with many squares.
\newblock {\em DMTCS}, 12(3):29--34, 2010.

\bibitem{Currie}
James~D. Currie.
\newblock Infinite ternary square-free words concatenated from permutations of
  a single word.
\newblock {\em Theor. Comput. Sci.}, 482:1--8, 2013.

\bibitem{Hall}
Marshall Hall, Jr.
\newblock Generators and relations in groups---{T}he {B}urnside problem.
\newblock In {\em Lectures on {M}odern {M}athematics, {V}ol. {II}}, pages
  42--92. Wiley, New York, 1964.

\bibitem{harju2}
Tero Harju.
\newblock Square-free words obtained from prefixes by permutations.
\newblock {\em Theoret. Comput. Sci.}, 429:128--133, 2012.

\bibitem{Harju}
Tero Harju.
\newblock A note on square-free shuffles of words.
\newblock In Juhani Karhum{\"a}ki, Arto Lepist{\"o}, and Luca~Q. Zamboni,
  editors, {\em WORDS}, volume 8079 of {\em Lecture Notes in Computer Science},
  pages 154--160. Springer, 2013.

\bibitem{Lothaire}
M.~Lothaire.
\newblock {\em Combinatorics on words}.
\newblock Cambridge Mathematical Library. Cambridge University Press,
  Cambridge, 1997.

\bibitem{ProdingerUrbanek}
P.~Prodinger and F.J. Urbanek.
\newblock Infinite 0-1-sequences without long adjacent identical blocks.
\newblock {\em Discrete Math.}, 28:277--289, 1979.

\bibitem{RampersadSW}
Narad Rampersad, Jeffrey Shallit, and Ming wei Wang.
\newblock Avoiding large squares in infinite binary words.
\newblock {\em Theor. Comput. Sci.}, 339(1):19--34, 2005.

\bibitem{Rizzi}
Romeo Rizzi and St{\'e}phane Vialette.
\newblock On recognizing words that are squares for the shuffle product.
\newblock In Andrei~A. Bulatov and Arseny~M. Shur, editors, {\em CSR}, volume
  7913 of {\em Lecture Notes in Computer Science}, pages 235--245. Springer,
  2013.

\end{thebibliography}

\newpage
\section*{Appendix}
We use the following square-free words that can be self-shuffled to get a square-free word
using the corresponding conducting sequence to create longer words:
{\small\addtolength{\jot}{-0.2em}
        \begin{align*}
            w_3 &= 012 & \beta_3 &= 0^{2}101^{2} \\
            w_4 &= 0120 & \beta_4 &= 0^{2}10^{2}1^{3} \\
            w_5 &= 01201 & \beta_5 &= 0^{3}1^{2}0101^{2} \\
            w_6 &= 010212 & \beta_6 &= 0^{4}10^{2}1^{5} \\
            w_7 &= 0102120 & \beta_7 &= 0^{6}1^{4}01^{3} \\
            w_8 &= 01021201 & \beta_8 &= 0^{6}1^{4}01^{2}01^{2} \\
            w_9 &= 010212012 & \beta_{9} &=  0^{6}1^{4}01^{2}0101^{2} \\
            w_{10} &= 0102120102 & \beta_{10} &=  0^{6}1^{4}01^{2}0^{3}1^{4} \\
            w_{11} &= 01210212021 & \beta_{11} &=  0^{4}1^{3}0^{5}10^{2}1^{7} \\
            w_{12} &= 010201202120 & \beta_{12} &=  0^{11}1^{10}01^{2} \\
            w_{13} &= 0102012021201 & \beta_{13} &=  0^{11}1^{10}0^{2}1^{3} \\
            w_{14} &= 01210201202102 & \beta_{14} &= 0^{2}10^{3}10^{2}1^{5}0^{4}1^{3}0^{2}1^{2}01^{2} \\
            w_{15} &= 012102012021201 & \beta_{15} &=  0^{2}10^{3}10^{2}1^{5}0^{4}1^{3}01^{2}0^{2}101^{2} \\
            w_{16} &= 0102012021012102 & \beta_{16} &=  0^{13}1^{4}01^{6}0^{2}1^{6} \\
            w_{17} &= 01020120210121020 & \beta_{17} &=  0^{13}1^{7}01^{2}01^{2}01^{2}01^{4} \\
            w_{19} &= 0102012021012102012 & \beta_{19} &=  0^{13}1^{4}01^{6}0^{2}1^{2}01^{3}01^{2}01^{2} \\
            w_{20} &= 01020120210121020102 & \beta_{20} &=  0^{13}1^{7}0^{5}1^{6}0^{2}1^{7} \\
            w_{21} &= 010201202101210201021 & \beta_{21} &=  0^{13}1^{4}01^{6}0^{7}1^{11} \\
            w_{26} &= 01020120210120102101202120 & \beta_{26} &=  0^{25}1^{24}01^{2}
        \end{align*}
    }
    We obtain longer square-free words having the same property, by applying compositions of
    the following square-free morphisms to them (their square-freeness is checked using
    Theorem \ref{Crochemore}):
{\small\addtolength{\jot}{-0.2em}
\begin{align*}
    \sigma_{1}(0) &= 1 & \sigma_{1}(1) &= 2 & \sigma_{1}(2) &= 0 \\
    \sigma_{2}(0) &= 2 & \sigma_{2}(1) &= 0 & \sigma_{2}(2) &= 1 \\
    \sigma_{3}(0) &= 1 & \sigma_{3}(1) &= 0 & \sigma_{3}(2) &= 2 \\
    \sigma_{4}(0) &= 2 & \sigma_{4}(1) &= 1 & \sigma_{4}(2) &= 0 \\
    \sigma_{5}(0) &= 0 & \sigma_{5}(1) &= 2 & \sigma_{5}(2) &= 1 \\
    \sigma_{6}(0) &= 0102012 & \sigma_{6}(1) &= 021012 & \sigma_{6}(2) &= 10212  \\
    \sigma_{7}(0) &= 0102012 & \sigma_{7}(1) &= 021012 & \sigma_{7}(2) &= 102010212 \\
    \sigma_{8}(0) &= 0102012 & \sigma_{8}(1) &= 0210201021012 & \sigma_{8}(2) &= 10212 \\
    \sigma_{9}(0) &= 0102 & \sigma_{9}(1) &= 01210120212 & \sigma_{9}(2) &= 012102010212  \\
    \sigma_{10}(0) &= 0102012 & \sigma_{10}(1) &= 021012 & \sigma_{10}(2) &= 1020102101210212 \\
    \sigma_{11}(0) &= 0 & \sigma_{11}(1) &= 102012021201021012 & \sigma_{11}(2) &= 10212021012  \\
    \sigma_{12}(0) &= 0 & \sigma_{12}(1) &= 102012021012 & \sigma_{12}(2) &= 102120210201021012  \\
    \sigma_{13}(0) &= 0102012 & \sigma_{13}(1) &= 0210201021012 & \sigma_{13}(2) &= 021201210212  \\
    \sigma_{14}(0) &= 0102012 & \sigma_{14}(1) &= 0210201021012 & \sigma_{14}(2) &= 0210201210212  \\
    \sigma_{15}(0) &= 0102 & \sigma_{15}(1) &= 012101201020120212 & \sigma_{15}(2) &= 012102010212  \\
    \sigma_{16}(0) &= 0102012 & \sigma_{16}(1) &= 02101210201021012 & \sigma_{16}(2) &= 02102010212  \\
    \sigma_{17}(0) &= 0 & \sigma_{17}(1) &= 102012021012102010212021012 & \sigma_{17}(2) &= 10201202120102120210201021012
\end{align*}
}
\newpage
In the following table, an entry $\sigma_{i_1 \cdot \ldots \cdot i_j}(w_k)$ is short-hand for
$\sigma_{i_1} \circ \cdots \circ \sigma_{i_j}(w_k)$: {\scriptsize\addtolength{\jot}{-0.2em}
\noindent\begin{multicols}{6}
\vspace*{-28pt}
\begin{align*}
w_{18} &= \sigma_{6}(w_{3})\\
w_{22} &= \sigma_{7}(w_{3})\\
w_{23} &= \sigma_{6\cdot 2}(w_{4})\\
w_{24} &= \sigma_{6\cdot 1}(w_{4})\\
w_{25} &= \sigma_{8}(w_{3})\\
w_{27} &= \sigma_{9}(w_{3})\\
w_{28} &= \sigma_{7\cdot 1}(w_{4})\\
w_{29} &= \sigma_{10}(w_{3})\\
w_{30} &= \sigma_{11}(w_{3})\\
w_{31} &= \sigma_{12}(w_{3})\\
w_{32} &= \sigma_{13}(w_{3})\\
w_{34} &= \sigma_{15}(w_{3})\\
w_{35} &= \sigma_{16}(w_{3})\\
w_{37} &= \sigma_{7\cdot 1}(w_{5})\\
w_{38} &= \sigma_{15}(w_{4})\\
w_{40} &= \sigma_{14}(w_{4})\\
w_{41} &= \sigma_{11\cdot 2}(w_{4})\\
w_{42} &= \sigma_{16}(w_{4})\\
w_{43} &= \sigma_{12\cdot 1}(w_{4})\\
w_{45} &= \sigma_{10\cdot 2}(w_{4})\\
w_{46} &= \sigma_{14\cdot 1}(w_{4})\\
w_{47} &= \sigma_{6\cdot 1}(w_{8})\\
w_{49} &= \sigma_{12\cdot 2}(w_{4})\\
w_{50} &= \sigma_{9\cdot 1}(w_{5})\\
w_{53} &= \sigma_{14}(w_{5})\\
w_{56} &= \sigma_{15}(w_{5})\\
w_{58} &= \sigma_{17}(w_{4})\\
w_{59} &= \sigma_{16}(w_{5})\\
w_{61} &= \sigma_{12\cdot 1}(w_{5})\\
w_{62} &= \sigma_{12}(w_{6})\\
w_{63} &= \sigma_{16\cdot 1}(w_{5})\\
w_{64} &= \sigma_{15\cdot 1}(w_{5})\\
w_{67} &= \sigma_{6\cdot 2}(w_{11})\\
w_{70} &= \sigma_{16}(w_{6})\\
w_{71} &= \sigma_{13}(w_{7})\\
w_{73} &= \sigma_{14}(w_{7})\\
w_{74} &= \sigma_{10\cdot 2}(w_{7})\\
w_{79} &= \sigma_{14\cdot 1}(w_{7})\\
w_{80} &= \sigma_{12\cdot 2}(w_{7})\\
w_{82} &= \sigma_{8}(w_{10})\\
w_{83} &= \sigma_{13\cdot 2}(w_{8})\\
w_{86} &= \sigma_{17\cdot 2}(w_{4})\\
w_{89} &= \sigma_{11\cdot 1}(w_{8})\\
w_{94} &= \sigma_{16}(w_{8})\\
w_{97} &= \sigma_{9\cdot 1}(w_{11})\\
w_{98} &= \sigma_{15\cdot 1}(w_{8})\\
w_{101} &= \sigma_{11\cdot 2}(w_{10})\\
w_{103} &= \sigma_{13}(w_{10})\\
w_{106} &= \sigma_{14}(w_{10})\\
w_{107} &= \sigma_{9\cdot 1}(w_{12})\\
w_{109} &= \sigma_{13\cdot 1}(w_{10})\\
w_{113} &= \sigma_{17\cdot 1}(w_{5})\\
w_{118} &= \sigma_{15\cdot 1}(w_{11})\\
w_{122} &= \sigma_{16\cdot 1}(w_{10})\\
w_{127} &= \sigma_{11\cdot 1}(w_{12})\\
w_{131} &= \sigma_{6\cdot 2\cdot 7}(w_{3})\\
w_{134} &= \sigma_{13\cdot 3}(w_{12})\\
w_{137} &= \sigma_{6\cdot 6\cdot 2}(w_{4})\\
w_{139} &= \sigma_{14}(w_{13})\\
w_{142} &= \sigma_{17}(w_{8})\\
w_{146} &= \sigma_{16\cdot 1}(w_{12})\\
w_{149} &= \sigma_{6\cdot 2\cdot 8}(w_{3})\\
w_{151} &= \sigma_{6\cdot 8}(w_{3})\\
w_{157} &= \sigma_{16\cdot 1}(w_{13})\\
w_{158} &= \sigma_{16}(w_{14})\\
w_{163} &= \sigma_{7\cdot 2\cdot 7}(w_{3})\\
w_{166} &= \sigma_{15\cdot 2}(w_{14})\\
w_{167} &= \sigma_{13}(w_{16})\\
w_{173} &= \sigma_{6\cdot 1\cdot 10}(w_{3})\\
w_{178} &= \sigma_{14\cdot 1}(w_{16})\\
w_{179} &= \sigma_{14}(w_{17})\\
w_{181} &= \sigma_{10}(w_{19})\\
w_{191} &= \sigma_{6\cdot 2\cdot 13}(w_{3})\\
w_{193} &= \sigma_{9\cdot 7}(w_{3})\\
w_{194} &= \sigma_{6\cdot 8}(w_{4})\\
w_{197} &= \sigma_{6\cdot 3\cdot 14}(w_{3})\\
w_{199} &= \sigma_{6\cdot 14}(w_{3})\\
w_{202} &= \sigma_{17}(w_{12})\\
w_{206} &= \sigma_{15\cdot 1}(w_{17})\\
w_{211} &= \sigma_{7\cdot 10}(w_{3})\\
w_{214} &= \sigma_{7\cdot 2\cdot 10}(w_{3})\\
w_{218} &= \sigma_{12\cdot 7}(w_{3})\\
w_{223} &= \sigma_{6\cdot 7\cdot 1}(w_{5})\\
w_{226} &= \sigma_{7\cdot 1\cdot 12}(w_{3})\\
w_{227} &= \sigma_{7\cdot 12}(w_{3})\\
w_{229} &= \sigma_{7\cdot 2\cdot 12}(w_{3})\\
w_{233} &= \sigma_{7\cdot 1\cdot 13}(w_{3})\\
w_{239} &= \sigma_{10\cdot 8}(w_{3})\\
w_{241} &= \sigma_{11\cdot 8}(w_{3})\\
w_{251} &= \sigma_{7\cdot 2\cdot 15}(w_{3})\\
w_{254} &= \sigma_{7\cdot 1\cdot 16}(w_{3})\\
w_{257} &= \sigma_{8\cdot 12}(w_{3})\\
w_{262} &= \sigma_{8\cdot 13}(w_{3})\\
w_{263} &= \sigma_{13\cdot 8}(w_{3})\\
w_{269} &= \sigma_{8\cdot 14}(w_{3})\\
w_{271} &= \sigma_{14\cdot 8}(w_{3})\\
w_{274} &= \sigma_{9\cdot 12}(w_{3})\\
w_{277} &= \sigma_{8\cdot 5\cdot 14}(w_{3})\\
w_{278} &= \sigma_{7\cdot 15}(w_{4})\\
w_{281} &= \sigma_{8\cdot 3\cdot 14}(w_{3})\\
w_{283} &= \sigma_{8\cdot 1\cdot 14}(w_{3})\\
w_{293} &= \sigma_{9\cdot 2\cdot 13}(w_{3})\\
w_{298} &= \sigma_{9\cdot 4\cdot 14}(w_{3})\\
w_{302} &= \sigma_{8\cdot 10}(w_{4})\\
w_{307} &= \sigma_{13\cdot 2\cdot 10}(w_{3})\\
w_{311} &= \sigma_{12\cdot 12}(w_{3})\\
w_{313} &= \sigma_{10\cdot 13}(w_{3})\\
w_{314} &= \sigma_{8\cdot 15}(w_{4})\\
w_{317} &= \sigma_{14\cdot 10}(w_{3})\\
w_{326} &= \sigma_{10\cdot 15}(w_{3})\\
w_{331} &= \sigma_{11\cdot 15}(w_{3})\\
w_{334} &= \sigma_{13\cdot 8}(w_{4})\\
w_{337} &= \sigma_{14\cdot 12}(w_{3})\\
w_{346} &= \sigma_{13\cdot 14}(w_{3})\\
w_{347} &= \sigma_{11\cdot 3\cdot 14}(w_{3})\\
w_{349} &= \sigma_{6\cdot 1\cdot 17}(w_{4})\\
w_{353} &= \sigma_{12\cdot 1\cdot 15}(w_{3})\\
w_{358} &= \sigma_{12\cdot 2\cdot 14}(w_{3})\\
w_{359} &= \sigma_{13\cdot 15}(w_{3})\\
w_{362} &= \sigma_{15\cdot 1\cdot 13}(w_{3})\\
w_{367} &= \sigma_{16\cdot 1\cdot 12}(w_{3})\\
w_{373} &= \sigma_{11\cdot 13}(w_{4})\\
w_{379} &= \sigma_{16\cdot 4\cdot 14}(w_{3})\\
w_{382} &= \sigma_{15\cdot 16}(w_{3})\\
w_{383} &= \sigma_{8\cdot 10\cdot 2}(w_{4})\\
w_{386} &= \sigma_{15\cdot 2\cdot 15}(w_{3})\\
w_{389} &= \sigma_{14\cdot 1\cdot 16}(w_{3})\\
w_{394} &= \sigma_{11\cdot 4\cdot 14}(w_{4})\\
w_{397} &= \sigma_{11\cdot 1\cdot 13}(w_{4})\\
w_{398} &= \sigma_{15\cdot 2\cdot 16}(w_{3})\\
w_{401} &= \sigma_{17\cdot 2}(w_{21})\\
w_{409} &= \sigma_{9\cdot 15\cdot 2}(w_{4})\\
w_{419} &= \sigma_{7\cdot 17}(w_{3})\\
w_{421} &= \sigma_{7\cdot 4\cdot 17}(w_{3})\\
w_{422} &= \sigma_{13\cdot 3\cdot 13}(w_{4})\\
w_{431} &= \sigma_{8\cdot 14\cdot 2}(w_{5})\\
w_{433} &= \sigma_{11\cdot 13\cdot 1}(w_{4})\\
w_{439} &= \sigma_{12\cdot 1\cdot 16}(w_{4})\\
w_{443} &= \sigma_{10\cdot 14\cdot 1}(w_{4})\\
w_{446} &= \sigma_{8\cdot 1\cdot 13}(w_{5})\\
w_{449} &= \sigma_{16\cdot 4\cdot 13}(w_{4})\\
w_{454} &= \sigma_{15\cdot 16}(w_{4})\\
w_{457} &= \sigma_{17\cdot 8}(w_{3})\\
w_{458} &= \sigma_{16\cdot 5\cdot 14}(w_{4})\\
w_{461} &= \sigma_{16\cdot 1\cdot 13}(w_{4})\\
w_{463} &= \sigma_{9\cdot 15\cdot 1}(w_{4})\\
w_{466} &= \sigma_{12\cdot 15\cdot 2}(w_{4})\\
w_{467} &= \sigma_{8\cdot 4\cdot 17}(w_{3})\\
w_{478} &= \sigma_{8\cdot 2\cdot 17}(w_{4})\\
w_{479} &= \sigma_{9\cdot 2\cdot 13}(w_{5})\\
w_{482} &= \sigma_{16\cdot 3\cdot 14}(w_{4})\\
w_{487} &= \sigma_{13\cdot 15\cdot 2}(w_{4})\\
w_{491} &= \sigma_{10\cdot 16\cdot 1}(w_{4})\\
w_{499} &= \sigma_{12\cdot 7}(w_{7})\\
w_{502} &= \sigma_{14\cdot 14\cdot 2}(w_{4})\\
w_{503} &= \sigma_{10\cdot 13\cdot 2}(w_{5})\\
w_{509} &= \sigma_{9\cdot 5\cdot 17}(w_{4})\\
w_{514} &= \sigma_{9\cdot 4\cdot 17}(w_{3})\\
w_{521} &= \sigma_{9\cdot 2\cdot 17}(w_{3})\\
w_{523} &= \sigma_{9\cdot 1\cdot 17}(w_{4})\\
w_{526} &= \sigma_{9\cdot 4\cdot 17}(w_{4})\\
w_{538} &= \sigma_{9\cdot 2\cdot 16}(w_{5})\\
w_{541} &= \sigma_{17\cdot 10}(w_{3})\\
w_{542} &= \sigma_{10\cdot 5\cdot 17}(w_{3})\\
w_{547} &= \sigma_{10\cdot 1\cdot 17}(w_{4})\\
w_{554} &= \sigma_{11\cdot 17}(w_{4})\\
w_{557} &= \sigma_{13\cdot 14\cdot 2}(w_{5})\\
w_{562} &= \sigma_{14\cdot 16\cdot 1}(w_{4})\\
w_{563} &= \sigma_{11\cdot 4\cdot 17}(w_{3})\\
w_{566} &= \sigma_{10\cdot 5\cdot 16}(w_{5})\\
w_{569} &= \sigma_{17\cdot 1\cdot 10}(w_{3})\\
w_{571} &= \sigma_{17\cdot 12}(w_{3})\\
w_{577} &= \sigma_{11\cdot 1\cdot 17}(w_{3})\\
w_{586} &= \sigma_{9\cdot 2\cdot 13}(w_{6})\\
w_{587} &= \sigma_{11\cdot 3\cdot 17}(w_{3})\\
w_{593} &= \sigma_{10\cdot 2\cdot 16}(w_{5})\\
w_{599} &= \sigma_{17\cdot 2\cdot 12}(w_{3})\\
w_{601} &= \sigma_{17\cdot 14}(w_{3})\\
w_{607} &= \sigma_{13\cdot 4\cdot 17}(w_{3})\\
w_{613} &= \sigma_{13\cdot 2\cdot 17}(w_{3})\\
w_{614} &= \sigma_{13\cdot 3\cdot 17}(w_{3})\\
w_{617} &= \sigma_{7\cdot 17\cdot 1}(w_{4})\\
w_{619} &= \sigma_{13\cdot 4\cdot 17}(w_{4})\\
w_{622} &= \sigma_{13\cdot 1\cdot 17}(w_{4})\\
w_{626} &= \sigma_{17\cdot 2\cdot 13}(w_{3})\\
w_{631} &= \sigma_{11\cdot 15\cdot 1}(w_{5})\\
w_{634} &= \sigma_{14\cdot 10}(w_{6})\\
w_{641} &= \sigma_{14\cdot 14\cdot 1}(w_{5})\\
w_{643} &= \sigma_{8\cdot 14\cdot 2}(w_{7})\\
w_{647} &= \sigma_{13\cdot 12\cdot 1}(w_{5})\\
w_{653} &= \sigma_{17\cdot 3\cdot 14}(w_{3})\\
w_{659} &= \sigma_{16\cdot 4\cdot 17}(w_{3})\\
w_{661} &= \sigma_{16\cdot 5\cdot 17}(w_{3})\\
w_{662} &= \sigma_{16\cdot 17}(w_{4})\\
w_{673} &= \sigma_{16\cdot 14\cdot 1}(w_{5})\\
w_{674} &= \sigma_{14\cdot 12}(w_{6})\\
w_{677} &= \sigma_{7\cdot 10\cdot 1}(w_{10})\\
w_{683} &= \sigma_{7\cdot 1\cdot 16}(w_{8})\\
w_{691} &= \sigma_{16\cdot 4\cdot 16}(w_{5})\\
w_{694} &= \sigma_{8\cdot 17\cdot 1}(w_{4})\\
w_{698} &= \sigma_{10\cdot 2\cdot 13}(w_{7})\\
w_{701} &= \sigma_{8\cdot 17}(w_{5})\\
w_{706} &= \sigma_{8\cdot 17\cdot 2}(w_{4})\\
w_{709} &= \sigma_{8\cdot 5\cdot 17}(w_{5})\\
w_{718} &= \sigma_{15\cdot 15\cdot 1}(w_{5})\\
w_{719} &= \sigma_{8\cdot 3\cdot 17}(w_{5})\\
w_{727} &= \sigma_{6\cdot 4\cdot 13}(w_{11})\\
w_{733} &= \sigma_{16\cdot 2\cdot 12}(w_{7})\\
w_{734} &= \sigma_{16\cdot 1\cdot 12}(w_{6})\\
w_{739} &= \sigma_{17\cdot 1\cdot 13}(w_{4})\\
w_{743} &= \sigma_{17\cdot 4\cdot 13}(w_{4})\\
w_{746} &= \sigma_{11\cdot 10\cdot 2}(w_{7})\\
w_{751} &= \sigma_{10\cdot 13\cdot 2}(w_{7})\\
w_{757} &= \sigma_{13\cdot 15}(w_{7})\\
w_{758} &= \sigma_{16\cdot 4\cdot 14}(w_{6})\\
w_{761} &= \sigma_{13\cdot 11\cdot 2}(w_{7})\\
w_{766} &= \sigma_{17\cdot 1\cdot 14}(w_{4})\\
w_{769} &= \sigma_{17\cdot 2\cdot 13}(w_{4})\\
w_{773} &= \sigma_{14\cdot 13}(w_{7})\\
w_{778} &= \sigma_{14\cdot 1\cdot 16}(w_{6})\\
w_{787} &= \sigma_{14\cdot 14}(w_{7})\\
w_{794} &= \sigma_{17\cdot 3\cdot 14}(w_{4})\\
w_{797} &= \sigma_{17\cdot 11\cdot 2}(w_{4})\\
w_{802} &= \sigma_{12\cdot 14\cdot 2}(w_{7})\\
w_{809} &= \sigma_{12\cdot 12\cdot 2}(w_{8})\\
w_{811} &= \sigma_{13\cdot 13\cdot 2}(w_{7})\\
w_{818} &= \sigma_{15\cdot 2\cdot 15}(w_{7})\\
w_{821} &= \sigma_{9\cdot 1\cdot 11}(w_{10})\\
w_{823} &= \sigma_{11\cdot 17\cdot 1}(w_{4})\\
w_{827} &= \sigma_{14\cdot 16}(w_{7})\\
w_{829} &= \sigma_{17\cdot 13\cdot 1}(w_{4})\\
w_{838} &= \sigma_{10\cdot 4\cdot 17}(w_{5})\\
w_{839} &= \sigma_{7\cdot 5\cdot 17}(w_{7})\\
w_{842} &= \sigma_{7\cdot 4\cdot 17}(w_{6})\\
w_{853} &= \sigma_{17\cdot 10\cdot 2}(w_{4})\\
w_{857} &= \sigma_{12\cdot 17\cdot 1}(w_{4})\\
w_{859} &= \sigma_{14\cdot 14\cdot 2}(w_{7})\\
w_{862} &= \sigma_{15\cdot 3\cdot 14}(w_{7})\\
w_{863} &= \sigma_{10\cdot 11\cdot 1}(w_{8})\\
w_{866} &= \sigma_{9\cdot 4\cdot 16}(w_{8})\\
w_{877} &= \sigma_{12\cdot 17\cdot 2}(w_{5})\\
w_{878} &= \sigma_{11\cdot 14\cdot 1}(w_{8})\\
w_{881} &= \sigma_{7\cdot 5\cdot 13}(w_{11})\\
w_{883} &= \sigma_{16\cdot 13\cdot 1}(w_{7})\\
w_{886} &= \sigma_{13\cdot 5\cdot 13}(w_{8})\\
w_{887} &= \sigma_{11\cdot 5\cdot 16}(w_{8})\\
w_{898} &= \sigma_{13\cdot 5\cdot 17}(w_{5})\\
w_{907} &= \sigma_{6\cdot 6\cdot 8}(w_{3})\\
w_{911} &= \sigma_{14\cdot 13\cdot 2}(w_{8})\\
w_{914} &= \sigma_{17\cdot 12\cdot 2}(w_{5})\\
w_{1031} &= \sigma_{10\cdot 10}(w_{11})\\
w_{1033} &= \sigma_{7\cdot 12\cdot 2}(w_{12})\\
w_{1039} &= \sigma_{6\cdot 5\cdot 17}(w_{10})\\
w_{1042} &= \sigma_{9\cdot 2\cdot 17}(w_{6})\\
w_{1046} &= \sigma_{15\cdot 2\cdot 10}(w_{10})\\
w_{1049} &= \sigma_{17\cdot 3\cdot 14}(w_{5})\\
w_{1051} &= \sigma_{6\cdot 6\cdot 10}(w_{3})\\
w_{1061} &= \sigma_{10\cdot 4\cdot 14}(w_{10})\\
w_{1063} &= \sigma_{14\cdot 9\cdot 1}(w_{11})\\
w_{1069} &= \sigma_{8\cdot 1\cdot 14}(w_{11})\\
w_{1082} &= \sigma_{17\cdot 10}(w_{6})\\
w_{1087} &= \sigma_{17\cdot 14\cdot 1}(w_{5})\\
w_{1091} &= \sigma_{10\cdot 5\cdot 17}(w_{7})\\
w_{1093} &= \sigma_{13\cdot 4\cdot 13}(w_{10})\\
w_{1094} &= \sigma_{16\cdot 10}(w_{10})\\
w_{1097} &= \sigma_{10\cdot 17\cdot 1}(w_{5})\\
w_{1103} &= \sigma_{10\cdot 10}(w_{12})\\
w_{1109} &= \sigma_{17\cdot 3\cdot 17}(w_{3})\\
w_{1114} &= \sigma_{17\cdot 4\cdot 17}(w_{4})\\
w_{1117} &= \sigma_{12\cdot 11\cdot 2}(w_{11})\\
w_{1373} &= \sigma_{6\cdot 7\cdot 12}(w_{3})\\
w_{1381} &= \sigma_{11\cdot 11\cdot 1}(w_{14})\\
w_{1382} &= \sigma_{6\cdot 6\cdot 15}(w_{4})\\
w_{1399} &= \sigma_{10\cdot 4\cdot 17}(w_{8})\\
w_{1402} &= \sigma_{16\cdot 13}(w_{12})\\
w_{1409} &= \sigma_{9\cdot 7\cdot 7}(w_{3})\\
w_{1418} &= \sigma_{13\cdot 13\cdot 2}(w_{12})\\
w_{1423} &= \sigma_{16\cdot 14}(w_{11})\\
w_{1427} &= \sigma_{6\cdot 7\cdot 13}(w_{3})\\
w_{1429} &= \sigma_{17\cdot 13\cdot 1}(w_{7})\\
w_{1433} &= \sigma_{16\cdot 2\cdot 12}(w_{11})\\
w_{1438} &= \sigma_{8\cdot 5\cdot 17}(w_{10})\\
w_{1439} &= \sigma_{11\cdot 13}(w_{14})\\
w_{1447} &= \sigma_{6\cdot 11\cdot 8}(w_{3})\\
w_{1451} &= \sigma_{16\cdot 1\cdot 12}(w_{11})\\
w_{1453} &= \sigma_{17\cdot 2\cdot 14}(w_{7})\\
w_{1454} &= \sigma_{6\cdot 6\cdot 14}(w_{4})\\
w_{1459} &= \sigma_{7\cdot 6\cdot 14}(w_{3})\\
w_{1466} &= \sigma_{15\cdot 3\cdot 14}(w_{11})\\
w_{1471} &= \sigma_{7\cdot 17\cdot 2}(w_{10})\\
w_{1478} &= \sigma_{7\cdot 17\cdot 1}(w_{11})\\
w_{1481} &= \sigma_{6\cdot 8\cdot 10}(w_{3})\\
w_{1483} &= \sigma_{17\cdot 11}(w_{8})\\
w_{1486} &= \sigma_{17\cdot 8}(w_{10})\\
w_{1487} &= \sigma_{12\cdot 4\cdot 17}(w_{8})\\
w_{1489} &= \sigma_{13\cdot 1\cdot 14}(w_{13})\\
w_{1493} &= \sigma_{16\cdot 3\cdot 14}(w_{11})\\
w_{1733} &= \sigma_{12\cdot 17\cdot 1}(w_{8})\\
w_{1741} &= \sigma_{11\cdot 6\cdot 10}(w_{3})\\
w_{1747} &= \sigma_{6\cdot 8\cdot 16}(w_{3})\\
w_{1753} &= \sigma_{11\cdot 5\cdot 10}(w_{19})\\
w_{1754} &= \sigma_{14\cdot 1\cdot 16}(w_{14})\\
w_{1759} &= \sigma_{13\cdot 16\cdot 2}(w_{14})\\
w_{1762} &= \sigma_{6\cdot 12\cdot 10}(w_{3})\\
w_{1766} &= \sigma_{15\cdot 7\cdot 7}(w_{3})\\
w_{1774} &= \sigma_{7\cdot 6\cdot 14}(w_{4})\\
w_{1777} &= \sigma_{9\cdot 17\cdot 2}(w_{10})\\
w_{1783} &= \sigma_{10\cdot 1\cdot 16}(w_{17})\\
w_{1787} &= \sigma_{11\cdot 2\cdot 13}(w_{17})\\
w_{1789} &= \sigma_{10\cdot 12}(w_{19})\\
w_{1801} &= \sigma_{13\cdot 14}(w_{16})\\
w_{1811} &= \sigma_{10\cdot 13\cdot 1}(w_{17})\\
w_{1814} &= \sigma_{6\cdot 11\cdot 8}(w_{4})\\
w_{1822} &= \sigma_{7\cdot 7\cdot 15}(w_{3})\\
w_{1823} &= \sigma_{7\cdot 12\cdot 8}(w_{3})\\
w_{1831} &= \sigma_{11\cdot 2\cdot 16}(w_{16})
\end{align*}
\end{multicols}
}

\end{document}